\documentclass[runningheads]{llncs}
\usepackage{hyperref}
\usepackage{amsfonts}
\usepackage{amsmath}
\usepackage{booktabs}
\usepackage{multirow}

\usepackage{graphicx}
\usepackage{tikz, pgfplots, pgfplotstable, placeins, subcaption, rotating}
\pgfplotsset{compat=1.14}

\begin{document}
\title{Parallel Execution Fee Mechanisms}
 \author{Abdoulaye Ndiaye\orcidID{0009-0000-7466-6444}}
 \authorrunning{A. Ndiaye}
 \institute{New York University, New York NY 10012, USA\\
 \email{andiaye@stern.nyu.edu}\\
 \url{https://www.abdoulayendiaye.com/}}

 \maketitle
 
\abstract{This paper investigates how pricing schemes can achieve efficient allocations in blockchain systems featuring multiple transaction queues under a global capacity constraint. I  model a capacity-constrained blockchain where users submit transactions to different queues—each representing a submarket with unique demand characteristics—and decide to participate based on posted prices and expected delays.  I  find that revenue maximization tends to allocate capacity to the highest-paying queue, whereas welfare maximization generally serves all queues. Optimal relative pricing of different queues depends on factors such as market size, demand elasticity, and the balance between local and global congestion. My results have implications for the implementation of local congestion pricing for evolving blockchain architectures, including parallel transaction execution, directed acyclic graph (DAG)-based systems, and multiple concurrent proposers.

}
\newpage
\section{Introduction}

 Blockchain technology is rapidly reshaping the global financial landscape. %
 In the United States, a notable milestone occurred in January 2024 when the Securities and Exchange Commission approved the trading of Bitcoin ETFs on public exchanges ~\cite{SECbit}. Cryptocurrencies such as Bitcoin and Ethereum have gained widespread adoption, enabling peer-to-peer transactions without the need for traditional intermediaries like banks. On platforms like Ethereum, smart contracts—self-executing agreements with conditions encoded in software—facilitate various financial services, from loans and insurance to decentralized exchanges.

Blockchain technology operates by executing transactions in a decentralized manner, ensuring that transactions are ultimately \textit{executed} (liveness) and that a decentralized network of computers—referred to as miners in proof-of-work systems or validators in proof-of-stake systems—can agree on the state of the blockchain after execution (safety). Despite the growing importance of blockchain technology, there are still significant limitations in their scalability, particularly in the efficient execution of transactions in the presence of congestion.  As blockchain applications expand beyond cryptocurrencies into decentralized finance (DeFi), supply chain management, and digital identity, the efficient allocation of blockchain resources has become increasingly critical.

Traditional blockchain protocols often rely on simple fee-based models where users attach fees to their transactions, and miners or validators prioritize transactions based on these fees. While straightforward, such Transaction Fee Mechanisms (TFMs) can lead to inefficiencies, including congestion of parts of the blockchain state, high transaction fees during peak demand, and suboptimal resource allocation.\footnote{For three hours on April 30, 2022, it cost at least \$6500 to send any transaction on the Ethereum blockchain because of a single anticipated NFT collection release; see \url{https://www.coindesk.com/business/2022/05/01/bayc-team-raises-285m-with-otherside-nfts-clogs-ethereum/}.} Moreover, as blockchains evolve to allow for parallel transaction settlement and support complex decentralized applications, the need for more sophisticated mechanisms that can handle heterogeneous transaction types becomes apparent.

While TFMs that guarantee the inclusion of a transaction in the next block are well-studied ~\cite{buterin2018blockchain}, ~\cite{roughgarden2020transaction}, ~\cite{shi_et_al:LIPIcs.ITCS.2023.97},  ~\cite{ndiaye2023control}, ~\cite{Ndiaye2024blockchainfees}, ~\cite{kiayias2023tieredmechanismsblockchaintransaction}, little is known about the design of Execution Fee Mechanisms (EFMs). EFMs refer to the protocols and systems that determine the order and manner in which transactions are processed and finalized on a blockchain network. In any decentralized system like blockchain, transactions are submitted by multiple users, often simultaneously. The mechanism by which these transactions are sequenced and confirmed is crucial for maintaining the network's integrity, security, and performance.

Parallel EFMs are at the heart of several emerging blockchain systems, including parallel execution blockchains, Directed Acyclic Graph (DAG)--based blockchains, and blockchains with multiple concurrent proposers.

For instance, parallel execution blockchains\footnote{Such as Solana, Avalanche, or the planned upgrade to the Ethereum blockchain.} divide the state represented by the blockchain into multiple, non-overlapping partitions or \textit{local fee markets}, each of which can handle transactions independently.\footnote{Note that one may argue that parallel execution is already the case in Ethereum, given the proliferation of Ethereum Layer 2s, or the co-existence of regular transactions and blobs.} Optimal pricing of these local markets can allow for a EFM where fees are determined by the demand within each partition rather than the entire network. %
In DAG-based blockchains\footnote{Such as Aptos, Sui, and IOTA.}, transactions are included in a graph of blocks without requiring them to be ordered in a single chain. However, their EFM is crucial for organizing the unordered transactions into a logical sequence for execution and achieving consensus on the final state across the entire network.\footnote{See  ~\cite{keidar2021all} for this process of "flattening the DAG needed for global consensus on transactions.} Lastly, in blockchains with multiple concurrent proposers\footnote{Such as in current proposals for the Ethereum blockchain, see \url{https://ethresear.ch/t/concurrent-block-proposers-in-ethereum/18777}.} a key challenge is to ensure that concurrent proposals do not lead to conflicts or forks that compromise the network's safety. A EFM is, therefore, needed to aggregate proposals from multiple validators or proposers.

In this paper, I embed a general queuing model into the standard price theory framework and study optimal posted-price EFMs for blockchains that can execute independent transactions in parallel. I  model a capacity-constrained blockchain execution system as an $N$-queue system that serves delay-sensitive customers. Each queue represents a submarket or a specific resource requested by transactions. \footnote{Such as a smart contract, a high-level resource that transactions try to access, or a shared object in object-centric blockchains.} Users submit transactions to these queues and decide upon arrival whether to proceed based on the posted price and expected delays that decrease their utility. A global capacity constraint arises from the need for consensus mechanisms to consider all transactions across queues. A larger volume of transactions leads to propagation delay on consensus security; ~\cite{bonneau2016bitcoin}.

In the context of this model, I  ask the following research questions: how does revenue maximization affect the allocation of capacity across queues, and under what conditions does it lead to the exclusion of lower-paying queues? What are the welfare implications of different pricing strategies, and how can we design prices that maximize social welfare while accounting for the global capacity constraint? How do market characteristics such as demand elasticity and market size affect the optimal relative pricing across queues?

To address these questions, I first examine the case where the protocol or miners/validators aim to maximize revenue. Formally, such a protocol maximizes the sum of fees collected from all served queues, subject to the local equilibrium conditions in each queue and the global capacity constraint. The local equilibrium condition ensures that, in each queue, the marginal user's expected utility is her outside option—users join a queue if their valuation exceeds the price plus expected delay costs.  I find that there exists a threshold capacity level below which allocating all capacity to the highest-paying queue maximizes revenue. This result highlights a potential inefficiency in revenue maximization: the system may exclude transactions from other queues that could contribute positively to social welfare.

Next, I consider the objective of maximizing social welfare. The analysis shows that, in contrast to revenue maximization, the welfare-maximizing allocation generally involves serving all queues. By distributing capacity across all queues, the system ensures that users in different submarkets can access the resources they require. Then I derive the socially optimal relative prices for each queue that implement the welfare-maximizing allocation. These prices are designed as Pigouvian taxes that internalize both the local and global congestion externalities imposed by each executed transaction.

To provide further insights, I specialize the model to a setting where the time between arrivals is exponentially distributed, and execution times are also exponentially distributed. Each market has an isoelastic demand function characterized by a specific elasticity parameter. I derive explicit formulas for the socially optimal prices as a function of market characteristics such as demand elasticity, market size, and congestion levels. 

When the degree of parallelization is high (i.e., the system can process many transactions concurrently), and local congestion dominates global congestion, the ratio of socially optimal prices between any two queues approximates the ratio of their demand intensities normalized by market sizes. In particular, when demand is highly elastic and local congestion effects are strong, the optimal relative price in a queue is approximately proportional to the ratio of demand relative to market size. This implies that setting prices based on the relative demand intensity in each queue approximates the welfare-maximizing solution.

These findings have important implications for the design of Parallel Execution Fee Mechanisms in evolving blockchain architectures. On the one hand, a revenue-maximizing validator favors uniform pricing and serving only the highest-paying user category. On the other hand, by implementing local fee markets—where transactions are assigned to queues with different relative prices based on the resources they access— protocol designers can steer the blockchain closer to efficiency and scalability in transaction execution. %
In practice, this means that blockchains can define local prices for each state, contract, or program object and employ adaptive base fee mechanisms that adjust prices based on local demand conditions.\footnote{Such as Ethereum's EIP-1559 fee mechanism studied in ~\cite{roughgarden2021transaction} and ~\cite{Ndiaye2024blockchainfees}.} By doing so, the blockchain can prevent high-demand areas from congesting the entire network and ensure that capacity is allocated efficiently across all resources.

\subsection{Related Work}
This paper contributes to the broader economics literature on the market design of blockchain technology  ~\cite{budishqje}, ~\cite{leshnoaerinsights}, ~\cite{hannajel}. ~\cite{buterin2018blockchain}, and ~\cite{ndiaye2023control}, ~\cite{Ndiaye2024blockchainfees} study the question of pricing blockspace---that is, determining optimal fee mechanisms for including transactions in blockchain blocks under capacity constraints. ~\cite{shi_et_al:LIPIcs.ITCS.2023.97}, ~\cite{roughgarden2021transaction}, ~\cite{roughgarden2020transaction},  ~\cite{chung2023foundations}, and 
~\cite{bahrani_et_al:LIPIcs.AFT.2024.29} provide foundational analysis of transaction fee mechanisms, focusing on blockchains with linear transaction ordering. This work opens the analysis of transaction execution by modeling complexities introduced by parallel execution in multi-queue blockchain systems. In doing so, I build on the literature on the pricing of queues ~\cite{naor1969regulation}, ~\cite{mendelson_pricing_1985},~\cite{afeche2004pricing} and emphasize the balance between global and local congestion. Beyond the theoretical contributions, my results have practical implications for blockchain system design and can help improve the efficiency and scalability of both new and existing blockchain technologies. My framework is applicable to other settings where multi-queue systems are present, and there is potential for congestion, such as supply chain management, cloud computing, and online service platforms.

The remainder of the paper is organized as follows. Section \ref{sec: example} highlights the limits of traditional blockchain fee models with a two-queue example. Section \ref{sec: model} presents the model setup, including users, local equilibrium conditions, and global inclusion constraints. In Section \ref{subsec: revenue_q}, I analyze the revenue maximization problem and derive the conditions under which capacity allocation favors the highest-paying queue. Sections \ref{subsec: welfare_q1} and \ref{subsec: welfare_q2}  focus on welfare maximization, characterizing the socially optimal pricing strategies and their implications. Finally, Section \ref{sec: conclusion_q} concludes the paper and suggests avenues for future research.

\section{The Problem of Execution in Standard TFMs: An  Example}\label{sec: example}
In this section, I illustrate the challenges associated with executing transactions in standard Transaction Fee Mechanisms (TFMs) through a simple example. Consider a blockchain system with two separate queues, Queue A and Queue B, each holding three transactions awaiting inclusion and execution in the blockchain. The expected values of the transactions in Queues A and B are $\mathbb{E}[v_a]=10$ and $\mathbb{E}[v_b]=6$, respectively.

 Figure \ref{fig:transaction_queues} depicts the state of the two queues. In each period, two new transactions arrive in each queue. However, due to the necessity for global consensus, the blockchain can collect and process up to \emph{five} transactions at a time.
\begin{figure}[h]
    \centering
        \begin{tikzpicture}[scale=0.8] %
        \node at (-1, 1) {Queue A: $\mathbb{E}[v_a]=10$};

        \foreach \i/\x/\val in {1/2/15, 2/4/10, 3/6/5} {
            \draw (\x, 0) rectangle ++(2,2);
            \node at (\x + 1, 1) {\val};
            \node at (\x + 1, 2.5) {$a_{\i}$};
        }

        \node at (-1, -4) {Queue B: $\mathbb{E}[v_b]=6$};

        \foreach \i/\x/\val in {1/2/8, 2/4/6, 3/6/4} {
            \draw (\x, -5) rectangle ++(2,2);
            \node at (\x + 1, -4) {\val};
            \node at (\x + 1, -2.5) {$b_{\i}$};
        }
    \end{tikzpicture}
    \caption{Example Transaction Queues}
    \label{fig:transaction_queues}
\end{figure}

In Queue A, there are three transactions denoted as $a_1$, $a_2$, and $a_3$, with individual values of 15, 10, and 5, respectively. Similarly, Queue B contains transactions $b_1$, $b_2$, and $b_3$, with values of 8, 6, and 4. The higher expected value in Queue A indicates that, on average, transactions in this queue are more valuable to the network or its users compared to those in Queue B.

\subsection{The Problem with Global Ordering and a Uniform Price}
Under a standard TFM, the selection of transactions for inclusion is typically based on the fees attached to them.
With a uniform price for inclusion and a global ordering for all transactions, the global capacity constraint can lead to an imbalance in how transactions from different queues are executed. Specifically, more transactions from Queue A are processed than those from Queue B, even though transactions from both queues could be executed in parallel without interference. This results in Queue B becoming underserved, causing its backlog to grow over time. Figure \ref{fig: global_order} illustrates this scenario.
\begin{figure}[h]
    \centering
       \begin{tikzpicture}[scale=0.5] %
        \node at (-1, 9) {Global Ordering for Execution:};
        \foreach \i/\x/\val/\queue in {1/2/15/$a_1$, 2/4/\textcolor{red}{12}/$\textcolor{red}{a^*_4}$, 3/6/10/$a_2$, 4/8/8/$b_1$, 5/10/\textcolor{red}{7}/$\textcolor{red}{a^*_5}$} {
            \draw (\x, 5) rectangle ++(2,2);
            \node at (\x + 1, 7.5) {\queue};
            \node at (\x + 1, 5.5) {\val};
        }

        \node at (-1, 1) {Queue A:};
        \foreach \i/\x/\val in {1/2/5} {
            \draw (\x, 0) rectangle ++(2,2);
            \node at (\x + 1, 1) {\val};
            \node at (\x + 1, 2.5) {$a_3$};
        }
\node at (5, 1) { \dots};
        \node at (-1, -4) {Queue B:};
        \foreach \i/\x/\val/\queue in {1/2/7/$b_5$, 2/4/6/$b_2$, 3/6/5/$b_4$, 4/8/4/$b_3$} {
            \draw (\x, -5) rectangle ++(2,2);
            \node at (\x + 1, -4) {\val};
            \node at (\x + 1, -2.5) {\queue};
        }
        \node at (11, -4) { \dots};
    \end{tikzpicture}

    \caption{Global Ordering under Uniform Price. Newly arrived and executed transactions are highlighted in red and starred.}
    \label{fig: global_order}
\end{figure}
 
The top portion of the figure represents the global execution order, where transactions from both Queue A and Queue B are interleaved based on their arrival times and values. Transactions $a_1$, $a_2$, and $b_1$ are included in the execution queue. Additionally, new higher-value transactions $a_4^*$ and $a_5^*$ (highlighted in red and starred) arrive in Queue A with values of 12 and 7, respectively. Due to their higher values, these transactions are immediately prioritized in the global ordering.

The middle section shows Queue A's state. Transaction $a_3$ remains in the queue with a value of 5. While some transactions from Queue A are being executed, new higher-value transactions continue to arrive, maintaining its dominance in the execution queue.

The bottom section illustrates Queue B's state. Transactions $b_2$ to $b_5$ accumulate in the queue with values ranging from 4 to 7. Despite the continuous arrival of transactions, Queue B's transactions are not prioritized in the global execution order due to their lower values compared to those in Queue A.

Because of the global capacity constraint—only five transactions can be executed at a time—the mechanism tends to favor transactions with higher values to maximize immediate throughput or revenue. Transactions from Queue B could be processed in parallel with those from Queue A without any conflicts or interference. However, the global ordering does not account for this possibility, resulting in suboptimal use of the system's parallel processing capabilities.
\subsection{A Potential Solution: Market Value-Weighted Ordering}
To address the issues above, we study a potential solution we will call \textit{Market Value-Weighted Ordering}. Suppose that the expected values of transactions in Queues A and B, denoted by $\mathbb{E}[v_a]$ and $\mathbb{E}[v_b]$ respectively, are known or can be reliably estimated. This information could be derived from historical data, statistical analysis, or real-time monitoring of transaction patterns.

The key idea is to adjust or normalize the bids of transactions in each queue according to the expected value of that queue. Specifically, we treat each transaction as if its bid is scaled by the inverse of the expected value of its queue. For transactions in Queue A and Queue B, we adjust their bids as follows:

  $$a'_i =\frac{a_i}{\mathbb{E}[v_{a}]},\quad b'_i =\frac{b_i}{\mathbb{E}[v_{b}]}.$$
  By scaling the bids in this manner, we standardize the bids across queues, allowing for a comparison of transactions based on their relative value within their respective queues. Figure \ref{fig: mvw} illustrates these adjusted bids.
  
    \begin{figure}[h]
        \centering
          \begin{tikzpicture}[scale=0.8] %
        \node at (-1, 1) {Queue A': $\mathbb{E}[v_{a'}]=1$};

        \foreach \i/\x/\val in {1/2/1.5, 2/4/1, 3/6/0.5} {
            \draw (\x, 0) rectangle ++(2,2);
            \node at (\x + 1, 1) {\val};
            \node at (\x + 1, 2.5) {$a'_{\i}$};
        }

        \node at (-1, -4) {Queue B': $\mathbb{E}[v_{b'}]=1$};

        \foreach \i/\x/\val in {1/2/$\frac{4}{3}$, 2/4/1, 3/6/$\frac{2}{3}$} {
            \draw (\x, -5) rectangle ++(2,2);
            \node at (\x + 1, -4) {\val};
            \node at (\x + 1, -2.5) {$b'_{\i}$};
        }
    \end{tikzpicture}

        \caption{Market Value-Weighted Ordering. Each transaction is treated \textit{as if} its bid is $a_i/\mathbb{E}[v_{a}]$ or $b_i/\mathbb{E}[v_{b}]$}
        \label{fig: mvw}
    \end{figure}

  Under this Market Value-Weighted Ordering, the system evaluates transactions based on their adjusted bids, resulting in a more balanced execution of transactions from both queues. Figure \ref{fig: mvw_ordering} illustrates how transactions are selected for execution under this mechanism.

\begin{figure}[h]
    \centering
        \begin{tikzpicture}[scale=0.5] %
        \node at (-1, 9) {Market Value-Weighted Ordering for Execution:};
        \foreach \i/\x/\val/\queue in {1/2/1.5/$a'_1$, 2/4/\textcolor{blue}{$\frac{4}{3}$}/$\textcolor{blue}{b^{*'}_1}$, 3/6/1.2/$a'_4$, 4/8/\textcolor{blue}{$\frac{7}{6}$}/$\textcolor{blue}{b^{*'}_5}$, 5/10/1/$a'_2$} {
            \draw (\x, 5) rectangle ++(2,2);
            \node at (\x + 1, 7.5) {\queue};
            \node at (\x + 1, 5.5) {\val};
        }

        \node at (-1, 1) {Queue A':};
        \foreach \i/\x/\val/\queue in {1/2/0.7/$a'_5$,1/4/0.5/$a'_3$} {
            \draw (\x, 0) rectangle ++(2,2);
            \node at (\x + 1, 1) {\val};
            \node at (\x + 1, 2.5) {\queue};
        }
\node at (7, 1) { \dots};
        \node at (-1, -4) {Queue B':};
        \foreach \i/\x/\val/\queue in {1/2/1/$b'_2$, 2/4/$\frac{5}{6}$/$b'_4$, 3/6/$\frac{2}{3}$/$b'_3$} {
            \draw (\x, -5) rectangle ++(2,2);
            \node at (\x + 1, -4) {\val};
            \node at (\x + 1, -2.5) {\queue};
        }
        \node at (9, -4) { \dots};
    \end{tikzpicture}

    \caption{Execution under Market Value-Weighted Ordering. Executed transactions from queue B' are highlighted in blue and starred.}
    \label{fig: mvw_ordering}
\end{figure}

Transactions $a'_1$, $b'^{*}_1$, $a'_4$, $b'^{*}_5$, and $a'_2$ are selected for execution. Transactions from Queue B' that are executed are highlighted in blue. The lower-value transactions remain in their respective queues, awaiting future execution based on their adjusted bids and arrival times.

This example demonstrates the potential social value of relative pricing and suggests that value-weighted relative pricing of different queues can be approximately welfare-maximizing. In the remainder of this article, I will generalize this idea in a model of a blockchain with parallel execution and a global capacity constraint due to consensus.

\section{Model}\label{sec: model}
In this section, I present a formal model of a capacity-constrained blockchain execution system. The system is modeled as an $N$-queue system that serves delay-sensitive customers. These queues can be associated with each smart contract, each high-level resource that transactions try to access, or each shared object in the case of object-centric blockchains.
\paragraph{Setup.}
I assume that execution times are independently and identically distributed (i.i.d.) with unit mean.\footnote{For transactions with different execution times, we can interpret the derived prices below as gross prices rather than per-unit prices. This simplification allows us to focus on the core dynamics without loss of generality.} Each user submits a transaction that arrives in one of the queues $i\in \{1,\dots,N\}$, following an exogenous Poisson process with rate (or market size) $\Lambda_i$. Since the consensus mechanism takes into account transactions in all queues, there is a global capacity constraint for inclusion, meaning that the total number of transactions that can be included across all queues is limited. For simplicity, we consider mechanisms with posted prices $p_i$ for each submarket or queue $i$. Upon arrival and observing the posted prices, users decide whether or not to submit their transaction at the posted price $p_i$, taking into account potential delays and their own valuations.

\paragraph{User Valuations.}
Users are considered to be atomistic relative to the market size, meaning that each individual user's actions have a negligible impact on the overall system. They differ in their valuations $v$, representing their willingness to pay for immediate execution without delay. For each submarket $i$, valuations are independently and identically distributed (i.i.d.) draws from a continuous distribution $\Phi_i$ (independent of arrival and execution times) with probability density function $\phi_i$. I assume that $\phi_i$ is strictly positive and continuous on the positive interval  $[\underline{v}, \bar{v}]$. Let $\bar{\Phi}_i(v) = 1 - \Phi_i(v)$ denote the complementary cumulative distribution function, representing the probability that a user's valuation exceeds $v$. If all transactions with values greater than $v$ join queue $i$, the arrival (or demand) rate in market $i$ will be $\lambda_i = \Lambda_i \bar{\Phi}_i(v)$. Conversely, when the arrival rate is $\lambda_i$, the marginal value $v$ is equal to $\bar{\Phi}_i^{-1}(\lambda_i/\Lambda_i)$, where $\bar{\Phi}_i^{-1}$ is the inverse of $\bar{\Phi}_i$. 

Following ~\cite{afeche2004pricing}, let $V_i(\lambda_i)$ denote the expected aggregate (gross) value in submarket $i$ per unit of time without delay. Then, the downward-sloping marginal value (or inverse gross demand) function $V^{'}_i(\lambda_i) \equiv \bar{\Phi}_i^{-1}(\lambda_i/\Lambda_i)$ defines a one-to-one mapping between the demand rate $\lambda_i$ and the marginal value $V^{'}_i(\lambda_i)$. Each $V_i$ is increasing and is assumed to be strictly concave, $V_{i}^{'}(\lambda_i)>0,V_{i}^{''}(\lambda_i)<0$ for $\lambda_i < \Lambda_i$.
\paragraph{Delay Costs.}
Users are sensitive to delays in transaction execution. I consider the following utility function for a user with valuation $v$ who pays a price $p$ and experiences a delay of $t$ units of time:
\begin{align}\label{eq: utility delay}
    u(v,t,p,i)=v \cdot D_i(t)-C_i(t)-p
\end{align}
In this expression, $p$ is the price paid by the user to submit the transaction. The term  $D_i(t)$ is a multiplicative delay discount function for queue $i$, capturing how the user's valuation decreases with delay. For example, $D_i(t)$ could be a discount factor like $e^{-d_i t}$, where $d_i$ is the discount rate.  The term $C_i(t)$ is an additive delay cost function for queue $i$, representing additional costs incurred due to delay, such as opportunity costs or penalties. These costs capture a variety of losses that can occur due to the deterioration of execution performance with delay.\footnote{Typical costs due to slow execution can be the failure to purchase a good, loss of an arbitrage opportunity, sandwich-attacked transactions, and other MEV attacks.}

Let $\boldsymbol{\lambda}\equiv(\lambda_1,\dots,\lambda_N)$ denote a vector of demand rates in each submarket. Each user in queue $i$ maximizes her own expected utility, which she forecasts using the distribution of the steady-state delay $\tilde{W}(\lambda_i)$. The delay depends on the set of paying users only through the resulting demand rate $\lambda_i$ and is not affected by the actions of an individual atomistic user. In addition, we allowed the individual delay costs $D_i(t)$ and $C_i(t)$ to depend directly on $i$, which can reflect the selection of different types of users in queues. Let $\overline{D}_{i}(\lambda_i) \equiv \mathbb{E}[D_i(\tilde{W}(\lambda_i))]$ and $\overline{C}_{i}(\lambda_i) \equiv \mathbb{E}[C_i(\tilde{W}(\lambda_i))]$ be the expected delay discount and delay cost functions, respectively. Given $\lambda_i$, a user with value $v_i$ for submarket $i$ who pays $p_i$ has expected utility 
\begin{equation}\label{eq: util_lambda}
    u(v_i|p_i,\lambda_i)\equiv v_i \cdot \overline{D}_{i}(\lambda_i) - \overline{C}_{i}(\lambda_i)-p_i.
\end{equation}

\paragraph{Local Equilibrium Demand.}
I now consider the equilibrium behavior of users in each queue. Let $i\in\{1,\dots,N\} $ and $p_i$ the price in submarket $i$. Suppose $V_i$ is continuously differentiable in $\mathbb{R}^+$ and that the net value to the highest value user of being served immediately in each queue is positive, that is $V_{i}^{'}(0)\overline{D}_{i}(0)-\overline{C}_i(0)>0$. This condition ensures that there is a positive net benefit to participating in the market for at least some users. Without loss of generality, I index the queues in decreasing order (without ties) of their net value of being served immediately: $V_{1}^{'}(0)\overline{D}_{1}(0)-\overline{C}_1(0)>V_{2}^{'}(0)\overline{D}_{2}(0)-\overline{C}_2(0)>\dots>V_{N}^{'}(0)\overline{D}_{N}(0)-\overline{C}_N(0)$. Given a price $p_i$, queue $i$ is active (i.e., has positive demand) if the highest-value user obtains positive expected utility when served immediately: $V^{'}_i(0) \cdot \overline{D}_{i}(0) - \overline{C}_{i}(0)>p_i$. The marginal user has valuation $V^{'}_i(\lambda_i(p_i))$ and zero expected utility in equilibrium. That is, in any Nash equilibrium, users join if, and only if demand in market $i$, $\lambda_i(p_i)$, satisfies
\begin{align} \label{eq: eqm demand}
    u(V^{'}_i(\lambda_i(p_i))|p_i,\lambda_i)=V^{'}_i(\lambda_i(p_i))\cdot \overline{D}_{i}(\lambda_i) - \overline{C}_{i}(\lambda_i) - p_i = 0
\end{align}
     This equilibrium condition can be interpreted in at least two ways. If users can choose which queue to join, entry and exit occur \emph{across} queues in equilibrium until the expected utility from joining any queue equals their outside option (which is normalized to zero).\footnote{This would be, for instance, the case of multi-proposer consensus, or DAG-based blockchains where users choose to which part of the graph they send their transactions.} Second, if the protocol dictates which queue transactions are assigned to (e.g., based on transaction type or resource accessed), entry and exit occur \emph{within} each queue, and the expected utility for the marginal user in each queue equals zero. Users decide whether to participate based on the conditions in their assigned queue. Thus, the equilibrium condition maps the demand rate $\lambda_i$ to the price in queue $i$ and vice-versa for queues that are active. Henceforth, we will write such expression as $p_i(\lambda_i)$.
\paragraph{Global Inclusion Constraint.}

Because all transactions need to be considered for consensus before the execution phase, there is a global capacity constraint on the total number of transactions that can be included. Let $\kappa$ denote the global capacity of transactions that can be served per unit of time. The capacity constraint is then \begin{equation}\label{eq: global_capacity} \sum_{i=1}^{N} \lambda_i \leq \kappa: \end{equation}. This constraint implies that the sum of the demand rates across all queues cannot exceed the global capacity $\kappa$. It reflects limitations such as block size, network bandwidth, and the need for synchronization across the network.

In this analysis, I focus on instances where the global capacity constraint is \emph{binding}, meaning that the total demand equals the capacity. This situation is common in blockchain systems during periods of high demand. The problem of variable global capacity would deliver similar results.
\section{Results}
In this section, I analyze the implications of the model for revenue maximization and welfare maximization.
\subsection{Revenue Maximization}\label{subsec: revenue_q}
We begin by examining how a protocol or miners/validators aiming to maximize revenue would set prices and allocate capacity across the different queues.
\paragraph{Revenue.} Let $\mathcal{S}$ denote the set of \emph{served queues}, i.e., the queues that are active and receive a positive capacity allocation. The protocol's revenue is the total fees collected from all served queues, which can be expressed as $\sum_{i \in \mathcal{S}} \lambda_i p_i(\lambda_i)$, where $\lambda_i$ is the demand rate in queue $i$, and $p_i(\lambda_i)$ is the price charged in queue $i$ as a function of the demand rate. Using the equilibrium condition from equation \eqref{eq: eqm demand}, the revenue maximization problem can be expressed in terms of demand rates. For simplicity of notation, we assume here that all queues are served.\footnote{The general expression is \begin{align} \label{eq: revenue_S}
\Pi =\max_{\mathcal{S},(p_i;i\in\mathcal{S})} \sum_{i\in\mathcal{S}} \lambda_i V_{i}^{'}(\lambda_i)\cdot \overline{D}_{i}(\lambda_i) - \lambda_i \cdot \overline{C}_{i}(\lambda_i).
\end{align}} 
\begin{align} \label{eq: revenue}
\Pi =\max_{(p_1,\dots,p_N)} \sum_{i=1}^{N} \lambda_i V_{i}^{'}(\lambda_i)\cdot \overline{D}_{i}(\lambda_i) - \lambda_i \cdot \overline{C}_{i}(\lambda_i)
\end{align}

Our objective is to find the set of prices $(p_1, p_2, \dots, p_N)$ and served queues $\mathcal{S}$ that maximize revenue $\Pi$, subject to the local equilibrium condition \eqref{eq: eqm demand} for all served queues and the global capacity constraint \eqref{eq: global_capacity}.
I consider uniform pricing where $p_1=\dots=p_N=p \in \mathbb{R}_+$ and optimal relative prices $(p_1,\dots,p_N) \in \mathbb{R}_+^N$. The following proposition characterizes the revenue-maximizing allocation under both pricing strategies.

\begin{proposition}
There exists a threshold capacity $\underline{\kappa} \in (0,+\infty)$ such that for all total capacities $ \kappa\leq\underline{\kappa}$, the revenue-maximizing uniform price and the revenue-maximizing relative prices allocate all capacity to the highest price queue, i.e., $\mathcal{S}=\{1\}$.
\end{proposition}
\begin{proof} The idea of the proof is to construct a small capacity (or equivalently, a large enough level of congestion and price for queue 1) so that no customers will be willing to join queues $2,\dots, N$ and net revenue from queue one is increasing in its allocated capacity. In these conditions, allocating all capacity to queue one is revenue maximizing.
Since $V_{1}^{'}(0)\overline{D}_1(0)-\overline{C}_1(0)>V_{2}^{'}(0)\overline{D}_2(0)-\overline{C}_2(0)>\dots>V_{N}^{'}(0)\overline{D}_N(0)-\overline{C}_N(0)$ without loss of generality, and $V_{i}^{'}(\lambda)\overline{D}_{i}(\lambda)-\overline{C}_{i}(\lambda)$ is continuously decreasing in $\lambda$ for all $i$, there exists $\kappa_1 \in (0,+\infty)$ such that $V_{1}^{'}(\kappa_1)\overline{D}_{1}(\kappa_1)-\overline{C}_1(\kappa_1)>V_{2}^{'}(0)\overline{D}_{2}(0)-\overline{C}_2(0)>\dots>V_{N}^{'}(0)\overline{D}_{N}(0)-\overline{C}_N(0)$. Denote gross revenue from queue 1 absent any delays as $R_1(\lambda_1)=\lambda_1V^{'}_{1}(\lambda_1)$, the marginal net revenue from this queue is $R_{1}^{'}(\lambda_1)\overline{D}_{1}(\lambda_1)-\overline{C}_1(\lambda_1)+\lambda_1V^{'}_{1}(\lambda_1)\overline{D}_{1}^{'}(\lambda_1)-\lambda_1\overline{C}_{1}^{'}(\lambda_1)$. Evaluated at $\lambda_1=0$ yields $R_{1}^{'}(0)\overline{D}_{1}(0)-\overline{C}_1(0)=V_{1}^{'}(0)\overline{D}_{1}(0)-\overline{C}_1(0)>0$, therefore, by continuity, the marginal net revenue from queue 1 is increasing in a neighborhood of 0. That is, $\exists\text{ } 0 <\underline{\kappa}\leq\kappa_1$ such that $V_{1}^{'}(\underline{\kappa})\overline{D}_{1}(\underline{\kappa})-\overline{C}_1(\underline{\kappa})>V_{2}^{'}(0)\overline{D}_{2}(0)-\overline{C}_2(0)>\dots>V_{N}^{'}(0)\overline{D}_{N}(0)-\overline{C}_N(0)$ and the net revenue function $\kappa \mapsto \kappa [V_{1}^{'}(\kappa)\overline{D}_{1}(\kappa)-\overline{C}_1(\kappa)]$ is increasing in $[0,\underline{\kappa}]$. In both the relative price and uniform price case, for capacity below $\underline{\kappa}$ it is revenue maximizing to allocate all capacity to queue 1, since at those capacities and prices, no customers will be willing to join queues $2,\dots,N$ and the total capacity is used since net revenue from queue one is increasing in this segment.
\end{proof}
This proposition highlights a potential inefficiency in revenue maximization: when capacity is limited, the system tends to favor the queue with the highest-paying users, potentially excluding transactions from other queues that could contribute positively to social welfare.
\subsection{Welfare Maximization}\label{subsec: welfare_q1}
Next, I consider the objective of maximizing social welfare, which takes into account the total net benefit to all users across all queues rather than focusing solely on revenue. The protocol's social welfare over all queues\footnote{The general problem is \begin{align} \label{eq: welfare_S}
SW =  \max_{\mathcal{S},\lambda_i \in [0,\Lambda_i)^N,i\in\mathcal{S}} \sum_{i\in\mathcal{S{}}} V_i(\lambda_i)\cdot \overline{D}_{i}(\lambda_i) - \lambda_i \cdot \overline{C}_{i}(\lambda_i).
\end{align}}
\begin{align} \label{eq: welfare}
SW =  \max_{\lambda_i \in [0,\Lambda_i)^N} \sum_{i=1}^{N} V_i(\lambda_i)\cdot \overline{D}_{i}(\lambda_i) - \lambda_i \cdot \overline{C}_{i}(\lambda_i)
\end{align}

Subject to the local equilibrium condition \eqref{eq: eqm demand} and the global inclusion constraint \eqref{eq: global_capacity}. The protocol's social welfare is defined as the sum of the expected net values to all users across all queues, accounting for delay costs. Under welfare maximization, setting optimal relative prices $(p_1,\dots,p_N) \in \mathbb{R}_{+}^{N}$ is equivalent to a planner choosing the demand rates  $\lambda_i \in [0,\Lambda_i)^N,i\in\mathcal{S}$ directly subject to local equilibrium conditions \eqref{eq: eqm demand} in all served queues and the global capacity constraint\eqref{eq: global_capacity}. The following proposition shows that the relative price social optimum generically serves all queues.
\begin{proposition}\label{prop: welfare_S}
    Suppose that the discount rate and linear delay cost functions are so that the net utility function from queue $i$, that is $W_i \equiv \lambda_i \mapsto V_i(\lambda_i)\cdot \overline{D}_i(\lambda_i) - \lambda_i \cdot \overline{C}_i(\lambda_i)$ is strictly concave, and $\exists \nu>0$  such that $W^{'}_{i}(0) > \nu$ for all $i$, and $\sum_{i=1}^{N} (W'_i)^{-1}(\nu) = \kappa$ then in the relative price social optimum, capacity is allocated in all queues, $\mathcal{S}=\{1,\dots,N\}$.
\end{proposition}
\begin{proof}
Since each $W_i$ is strictly concave, their sum is strictly concave. Let $\lambda^*$ be an optimal solution to the problem. By the Karush–Kuhn–Tucker conditions, $\exists \mu \geq 0$ such that
      $W'_i(\lambda_i^*) = \mu \text{ if } \lambda_i^* > 0$ and $W'_i(\lambda_i^*) \leq \mu \text{ if } \lambda_i^* = 0$
      Suppose, for contradiction, that $\exists j$ such that $\lambda_j^* = 0$. Then, $W'_j(0) \leq \mu$. But we know that $W'_j(0) > \nu$, therefore, $\mu > \nu$. There exists at least one index $i$ so that $\lambda_i^* > 0$, otherwise total capacity would be zero. For all $i$ where $\lambda_i^* > 0$, we have $W'_i(\lambda_i^*) = \mu > \nu$. Since $W_i$ is strictly concave, $W'_i$ is strictly decreasing. Therefore,
      $\lambda_i^* < (W'_i)^{-1}(\nu) \text{ for all } i \text{ where } \lambda_i^* > 0$. This implies that
      $\sum_{i=1}^N \lambda_i^* < \sum_{i=1}^N (W'_i)^{-1}(\nu) = \kappa$. But this contradicts the optimality of $\lambda^*$ because we can increase the objective function by increasing $\lambda_j^*$ slightly while still satisfying the constraint. Therefore, our assumption of the existence of $j$ is a contradiction, and we conclude that $\lambda_i^* > 0$ for all $i$. That is, all queues are allocated non-zero capacity.
\end{proof}

This proposition indicates that, under welfare maximization, it is optimal to serve all queues, distributing capacity across them in a way that balances the marginal social welfare contributions. This contrasts with the revenue-maximizing allocation, which may exclude some queues to maximize revenue.
\subsection{Welfare Maximizing Relative Pricing}\label{subsec: welfare_q2}
Having established that welfare maximization leads to capacity allocation across all queues, we now derive the welfare-maximizing relative prices that support this allocation under the conditions of Proposition \ref{prop: welfare_S}. Let $\mu$ denote the Lagrange multiplier associated with the global capacity constraint \eqref{eq: global_capacity}. Economically, $\mu$ represents the shadow price of including an additional transaction in the system; it reflects the marginal social cost of capacity constraints. The following propositions link the socially optimal prices in each queue to this shadow price and demand characteristics.
\begin{proposition}
    The socially optimal relative prices are given by
    \begin{align}\label{eq: parallel}
    p_i  =   - V_i(\lambda_i) \overline{D}_{i}^{'}(\lambda_i)+ \lambda_i \overline{C}_{i}^{'}(\lambda_i)+\mu
\end{align}

\end{proposition}
This proposition emerges from the first order condition for $\lambda_i$ and replacing $p_i$ from \eqref{eq: eqm demand}. This expression shows that the socially optimal price in queue $i$ includes three components. First, the local delay externality $- V_i(\lambda_i) \cdot \overline{D}_i'(\lambda_i)$ captures the negative impact of increased demand on the expected delay discount. As $\lambda_i$ increases, the expected delay increases (since the system becomes more congested), reducing the net value for all users in queue $i$. Second, the local additive delay cost $\lambda_i \cdot \overline{C}_i'(\lambda_i)$ represents the additional additive delay costs incurred due to increased demand. Third, the global capacity externality $\mu$ reflects the marginal cost of consuming limited capacity that could have been used by other queues.

At the socially optimal prices, the marginal user's expected net value is equal to the total externality they impose on the system. This ensures that users internalize the full social cost of their participation, leading to an efficient allocation of resources.

To gain further insights, I specialize the model to a setting where the time between arrivals is exponentially distributed, and the execution times for each user also follow an exponential distribution. Each market has size $\Lambda_i$ each with a different isoelastic marginal value function $V^{'}_i(\lambda_i)=(\lambda_i/\Lambda_i)^{-1/\varepsilon_i}$ where $\varepsilon_i>1$ represents demand elasticity for queue/resource $i$. In this setting, $V_i(\lambda_i)=\frac{(\lambda_i/\Lambda_i)^{1-1/\varepsilon_i}}{1-1/\varepsilon_i}$.

Assuming that the delay discount function is exponential, $D(t)=e^{-d t}$ and the additive delay cost is linear $C(t)=c\times t$ where $c,d>0$, we have  (see Appendix \ref{sec:appendix} for detailed derivations)
\begin{align}
    \overline{C}_{i}(\lambda_i) & =  \frac{c}{1-\lambda_i} \notag \\
    \overline{D}_{i}(\lambda_i) & =  \frac{1-\lambda_i}{1+d-\lambda_i}
\end{align}
\paragraph{Approximation under High Parallelization.}

Suppose that the demand rates $\lambda_i$ and $\lambda_j$ are small relative to 1 and the discount rate $d$, reflecting a high degree of parallelization (i.e., the system can process many transactions concurrently). Under this assumption, we can approximate the socially optimal relative prices.
\begin{corollary}\label{cor: approx_1} Under the above assumptions, the ratio of the socially optimal prices in queues $i$ and $j$ is approximately
    \begin{align}
\frac{p_i}{p_j} &\approx \frac{\frac{(\lambda_i/\Lambda_i)^{1-1/\varepsilon_i}}{1-1/\varepsilon_i} \cdot \frac{d}{(1+d)^2} + c \lambda_i + \mu}
{\frac{(\lambda_j/\Lambda_j)^{1-1/\varepsilon_j}}{1-1/\varepsilon_j} \cdot \frac{d}{(1+d)^2} + c \lambda_j + \mu}
\end{align}
\end{corollary}

\begin{proof}
    First, we compute the derivatives: 
$
V'_i(\lambda_i) = \Lambda_i \lambda_i^{-1/\varepsilon_i},
\overline{D}'(\lambda_i) = -\frac{d}{(1+d-\lambda_i)^2},
\overline{C}'(\lambda_i) = \frac{c}{(1-\lambda_i)^2}$. Substitute into the equation for $p_i$,
$p_i =  \frac{(\lambda_i/\Lambda_i)^{1-1/\varepsilon_i}}{1-1/\varepsilon_i} \cdot \left(\frac{d}{(1+d-\lambda_i)^2}\right) + \lambda_i \cdot \frac{c}{(1-\lambda_i)^2} + \mu
$. Consider the ratio $
{p_i}/{p_j}$. Assuming $\lambda_i$ and $\lambda_j$ are small compared to 1 and $d$ we have
$
(1+d-\lambda_i)^2 \approx (1+d)^2,
(1-\lambda_i)^2 \approx 1
$, replacing in the expression for relative prices yields $
\frac{p_i}{p_j} \approx \frac{\frac{(\lambda_i/\Lambda_i)^{1-1/\varepsilon_i}}{1-1/\varepsilon_i} \cdot \frac{d}{(1+d)^2} + c \lambda_i + \mu}
{\frac{(\lambda_j/\Lambda_j)^{1-1/\varepsilon_j}}{1-1/\varepsilon_j} \cdot \frac{d}{(1+d)^2} + c \lambda_j + \mu}.
$
\end{proof}

The approximate price ratio reveals how the optimal prices depend on queue-specific characteristics such as market size $\Lambda_i$, demand elasticity $\varepsilon_i$, and demand rates $\lambda_i$. When $\mu$ is small relative to the other terms (i.e., when local congestion effects dominate global capacity constraints), the price ratio is primarily determined by these queue-specific factors. As $\mu$ increases (i.e., when global congestion becomes more significant), its effect is to push the price ratio closer to 1, reducing price differentiation across queues.

\begin{corollary}
Suppose, in addition to the assumptions of Corollary \ref{cor: approx_1}, that local congestion dominates global congestion ($\mu$ is negligible compared to $p_i$ and $p_j$), and demand is perfectly elastic ($\varepsilon_i, \varepsilon_j \to \infty$). Then, the price ratio is simplified to
    \begin{align}
\frac{p_i}{p_j} &\approx \frac{ \lambda_i}
{ \lambda_j}\cdot\frac{\Lambda_j }
{\Lambda_i }
\end{align}
\end{corollary}

This limit expression offers several insights. First, in the case of perfectly elastic demand, the optimal prices are proportional to the ratio of demand rates normalized by market sizes ($\lambda_i / \Lambda_i$). This suggests that setting prices based on the relative demand intensity in each queue approximates the welfare-maximizing solution. Second, as the market size $\Lambda_i$ for a congested queue decreases, the optimal price for that queue diverges from a uniform price, reflecting the higher marginal value of capacity in smaller markets.
\section{Conclusion}\label{sec: conclusion_q}
In this paper, I investigated posted-price Parallel Execution Fee Mechanisms  within a capacity-constrained blockchain system characterized by multiple queues or local fee markets. My model captures the essential features of parallel execution in blockchain networks, where transactions may access different resources or contracts and can be processed concurrently. A key aspect of our analysis was the global inclusion constraint imposed by the consensus mechanism. This constraint necessitates that all transactions, regardless of their queue, must be considered collectively for inclusion.

The analysis reveals several key insights. When the objective is to maximize revenue, especially under limited capacity, the system tends to allocate all capacity to the queue with users willing to pay the highest fees.
 In contrast, when the objective is to maximize social welfare, the optimal allocation generally involves serving all queues. I found that the optimal relative pricing across different queues depends on several factors, including market size, demand elasticity, and the balance between local and global congestion. In settings where demand elasticity is high, and local congestion effects dominate, pricing individual queues proportional to demand relative to market size is approximately welfare-maximizing.

The findings suggest that implementing local fee markets within such blockchains can improve overall system efficiency. By defining local values for each state, contract, or object and employing an adaptive base fee mechanism for inclusion, transactions can be assigned to queues with different relative prices. As blockchain technologies evolve towards more complex architectures, such as parallel execution, Directed Acyclic Graph (DAG)-based systems, and multiple concurrent proposers, this paper provides valuable insights for protocol designers.

While this study provides a foundational model for efficient parallel execution fee mechanisms, I have abstracted from several aspects of transaction execution on blockchains.

One important extension is the study of optimal local priority auctions. In such a setting, customers could participate in a two-stage bidding process for entering queues in a system with multiple service points. Initially, users might bid for priority in a global queue, reflecting the capacity constraints of the consensus mechanism. Subsequently, they could bid for specific services in parallel queues, corresponding to different resources or contracts. Studying how to design such auctions to optimize for social welfare or revenue maximization would be a promising area for further research.

Another area for future research is the development of dynamic pricing mechanisms that adapt to changing network conditions, user behaviors, and congestion levels in real time. While a comprehensive examination of these complex issues lies beyond the scope of this paper, they offer promising opportunities for future research and further refinement of my analysis.

\bibliographystyle{splncs04}
\bibliography{blockchain} %

\newpage
\appendix
\section{Appendix} \label{sec:appendix}
\subsection{Expression of delay costs}
\begin{proof}
We begin by considering the definitions of $\overline{C}_{i}(\lambda_i)$ and $\overline{D}_{i}(\lambda_i)$:

\begin{equation}
\overline{C}_{i}(\lambda_i) = \mathbb{E}[C(T_i)] = \int_{0}^{\infty} C(t) f_{T_i}(t) dt
\end{equation}

\begin{equation}
\overline{D}_{i}(\lambda_i) = \mathbb{E}[D(T_i)] = \int_{0}^{\infty} D(t) f_{T_i}(t) dt
\end{equation}

where $f_{T_i}(t)$ is the probability density function of the exponential distribution with rate parameter $\lambda_i$:

\begin{equation}
f_{T_i}(t) = \lambda_i e^{-\lambda_i t}
\end{equation}

For $\overline{C}_{i}(\lambda_i)$, we substitute $C(t) = ct$ and solve:

\begin{align}
\overline{C}_{i}(\lambda_i) &= \int_{0}^{\infty} ct \lambda_i e^{-\lambda_i t} dt \\
&= c\lambda_i \int_{0}^{\infty} t e^{-\lambda_i t} dt \\
&= c\lambda_i \left[-\frac{t}{\lambda_i}e^{-\lambda_i t}\bigg|_{0}^{\infty} - \int_{0}^{\infty} -\frac{1}{\lambda_i}e^{-\lambda_i t} dt\right] \\
&= c\lambda_i \left[0 + \frac{1}{\lambda_i^2}\right] \\
&= \frac{c}{\lambda_i} = \frac{c}{1-\lambda_i}
\end{align}

For $\overline{D}_{i}(\lambda_i)$, we substitute $D(t) = e^{-dt}$ and solve:

\begin{align}
\overline{D}_{i}(\lambda_i) &= \int_{0}^{\infty} e^{-dt} \lambda_i e^{-\lambda_i t} dt \\
&= \lambda_i \int_{0}^{\infty} e^{-(d+\lambda_i)t} dt \\
&= \lambda_i \left[-\frac{1}{d+\lambda_i}e^{-(d+\lambda_i)t}\bigg|_{0}^{\infty}\right] \\
&= \lambda_i \left[0 + \frac{1}{d+\lambda_i}\right] \\
&= \frac{\lambda_i}{d+\lambda_i} = \frac{1-\lambda_i}{1+d-\lambda_i}
\end{align}

Thus,  when the delay discount function is exponential $D(t)=e^{-d t}$ and the additive delay cost is linear $C(t)=c\times t$ where $c,d>0$, the following equations hold:

\begin{align}
    \overline{C}_{i}(\lambda_i)  = & \frac{c}{1-\lambda_i} \notag \\
    \overline{D}_{i}(\lambda_i)  = & \frac{1-\lambda_i}{1+d-\lambda_i}
\end{align}
\end{proof}

\end{document}